\documentclass[11pt]{article}

\usepackage[papersize={8.5in,11in},top=1in,left=1in,right=1in,bottom=1in]{geometry}
\usepackage{ifthen}
\usepackage{xspace}
\usepackage{amsmath}
\usepackage{bm}
\usepackage[section]{definitions}
\newcommand{\MECH}{\ensuremath{\mathcal{M}}\xspace}
\newcommand{\vc}[1]{\ensuremath{\bm{#1}}\xspace}

\newcommand{\ALLCANDS}{\ensuremath{X}\xspace}
\newcommand{\ALLVOTERS}{\ensuremath{V}\xspace}
\newcommand{\ALLPERMS}{\ensuremath{\mathcal{S}_n}\xspace}
\newcommand{\PERMMAT}{\ensuremath{\mathcal{P}}\xspace}

\newcommand{\PERM}[1][]{\ensuremath{%
    \ifthenelse{\equal{#1}{}}{\pi}{\pi_{#1}}}\xspace}
\newcommand{\Perm}[2][]{\ensuremath{\PERM[#1](#2)}\xspace}
\newcommand{\PermInv}[2][]{\ensuremath{\PERM[#1]^{-1}(#2)}\xspace}
\newcommand{\PERMP}[1][]{\ensuremath{%
    \ifthenelse{\equal{#1}{}}{\pi'}{\pi'_{#1}}}\xspace}
\newcommand{\PermP}[2][]{\ensuremath{\PERMP[#1](#2)}\xspace}

\newcommand{\SCRULE}{\ensuremath{f}\xspace}
\newcommand{\SCRule}[1]{\ensuremath{\SCRULE(#1)}\xspace}

\newcommand{\Cost}[1]{\ensuremath{C(#1)}\xspace}
\newcommand{\Distortion}[1]{\ensuremath{\rho(#1)}\xspace}

\newcommand{\DIST}{\ensuremath{d}\xspace}
\newcommand{\Dist}[2]{\ensuremath{\DIST(#1,#2)}\xspace}
\newcommand{\Legal}[2]{\ensuremath{#1 \sim #2}\xspace}

\newcommand{\OPT}[1][]{\ensuremath{%
    \ifthenelse{\equal{#1}{}}{x^*}{x^*_{#1}}}\xspace}
\newcommand{\WINNER}{\ensuremath{w}\xspace}

\newcommand{\PARTS}[1][]{\ensuremath{%
    \ifthenelse{\equal{#1}{}}{M}{M_{#1}}}\xspace}
\newcommand{\PSet}[1]{\ensuremath{\Pi_{#1}}\xspace}

\newcommand{\MSG}{\ensuremath{\mu}\xspace}

\newcommand{\IC}{\ensuremath{\nu}\xspace}

\begin{document}

\title{Communication, Distortion, and Randomness in Metric Voting}
%\author{Anonymous Author(s)}
\author{David Kempe \\ University of Southern California}
\date{}

\maketitle

\begin{abstract}
  In distortion-based analysis of social choice rules over metric spaces, one assumes that all voters and candidates are jointly embedded in a common metric space. Voters rank candidates by non-decreasing distance. The mechanism, receiving only this ordinal (comparison) information, aims to nonetheless select a candidate approximately minimizing the sum of distances from all voters to the chosen candidate. It is known that while the Copeland rule and related rules guarantee distortion at most 5, many other standard voting rules, such as Plurality, Veto, or $k$-approval, have distortion growing unboundedly in the number $n$ of candidates.

An advantage of Plurality, Veto, or $k$-approval with small $k$ is that they require less communication from the voters; all deterministic social choice rules known to achieve constant distortion require voters to transmit their complete ranking of all candidates. This motivates our study of the tradeoff between the distortion and the amount of communication in deterministic social choice rules.

We show that any one-round deterministic voting mechanism in which each voter communicates only the candidates she ranks in a given set of $k$ positions must have distortion at least $\frac{2n-k}{k}$; we give a mechanism achieving an upper bound of $O(n/k)$, which matches the lower bound up to a constant.
For more general communication-bounded voting mechanisms, in which each voter communicates $b$ bits of information about her ranking, we show a slightly weaker lower bound of $\Omega(n/b)$ on the distortion.

For randomized mechanisms, the situation looks much brighter: it is known that Random Dictatorship achieves expected distortion strictly smaller than 3, almost matching a lower bound of $3-\frac{2}{n}$ for any randomized mechanism that only receives each voter's top choice. We close this gap, by giving a simple randomized social choice rule which only uses each voter's first choice, and achieves expected distortion $3-\frac{2}{n}$.

\end{abstract}

\section{Introduction} \label{sec:introduction}
In \emph{voting} or \emph{social choice},
there is a set of $n$ alternatives
(such as political candidates or courses of action)
from which a group
(such as a country or an organization)
wants to select a winner.\footnote{%
  A large and important part of the literature studies the
  goal of choosing a complete consensus ranking of all candidates;
  we will not study this alternative goal here,
  and therefore identify social choice with the selection of a single
  winner.}
Each voter submits a \emph{ranking} (or \emph{preference order})
of the candidates,
and the \emph{mechanism} (or \emph{social choice rule})
chooses a winner based on these submitted rankings.

Many different social choice rules have been proposed,
and it is an important question how to compare them.
One fruitful and long line of work,
dating back at least to the correspondence of Borda and Condorcet
\cite{borda:elections,condorcet:essay},
formulates axioms that a social choice rule ``should'' satisfy;
one can then compare social choice rules by which or how many of these
axioms they satisfy \cite{BCULP:social-choice}.
Unfortunately, many results in this area are impossibility results,
most notably Arrow's result for producing a consensus ranking
\cite{arrow:social-choice} and the
Gibbard-Satterthwaite Theorem ruling out truthful voting rules
with minimal additional properties
\cite{gibbard:manipulation,satterthwaite:voting}.

An alternative to the axiomatic approach is to consider social choice
as an \emph{optimization problem} with the goal of selecting the
``best'' candidate for the population
\cite{BCHLPS:utilitarian:distortion,caragiannis:procaccia:voting,procaccia:approximation:gibbard,procaccia:rosenschein:distortion}.
A natural way to express the notion of ``best'' is to assume that each
voter has a utility (or cost) for each candidate;
the mechanism's goal is to optimize the aggregate
(e.g., average or median) utility or cost of all voters.
However, as articulated in
\cite{boutilier:rosenschein:incomplete,anshelevich:bhardwaj:postl},
the social choice rule has to optimize with crucial information missing:
a voter can only communicate her\footnote{%
  For consistency and clarity, we will always refer to voters using
  female and candidates using male pronouns.}
\emph{ranking} according to the utility/cost.
In other words, the mechanism receives only \emph{ordinal} information ---
which candidate is preferred over which other candidate ---
even though it needs to optimize a \emph{cardinal} objective function.
From an optimization perspective, this means that the mechanism
should simultaneously optimize over all utility/cost functions that
are \emph{consistent} with the reported rankings,
in that they would give rise to the observed rankings.
The worst-case ratio (over all cost/utility functions)
between the mechanism's cost/utility and that of the optimum candidate
for the specific function is called the mechanism's
\emph{distortion}.
(Formal definitions of all concepts and terms are given in
Section~\ref{sec:preliminaries}.)

In applying this general framework, an important question is what
class of cost/utility functions to consider.
A natural approach was suggested in
\cite{anshelevich:bhardwaj:postl} (see also the expanded/improved
journal version \cite{anshelevich:bhardwaj:elkind:postl:skowron}
and general overview \cite{anshelevich:ordinal}):
all candidates and voters are jointly embedded in a metric space,
and the cost of voter $v$ for candidate $x$ is their metric distance
\Dist{v}{x}.
The assumption that voters rank candidates by non-decreasing distance
in a latent space dates back to earlier work on so-called
\emph{single-peaked preferences}
\cite{black:rationale,black:committees-elections,downs:democracy,moulin:single-peak,merrill:grofman,barbera:gul:stacchetti,barbera:social-choice},
though much of the earlier work focuses on the special case when the
metric is the line.
% enelow:hinich:spatial-theory,roemer:political-competition
% richards:richards:mckay
Using the framework of distortion and metric costs,
\cite{anshelevich:bhardwaj:postl,anshelevich:bhardwaj:elkind:postl:skowron}
show a remarkable separation.
While many commonly used voting rules (such as Plurality, Veto,
$k$-approval, Borda count) have either unbounded distortion 
or distortion linear in the number $n$ of candidates,
and indeed all score-based rules have distortion $\omega(1)$
(in terms of the number of candidates),
\emph{uncovered-set rules} have distortion at most 5.
To describe uncovered-set rules, 
consider a tournament graph $G$ on the $n$ candidates which
contains the directed edge $(x,y)$ iff at least as many voters prefer
$x$ to $y$ as vice versa.
The uncovered set of $G$ is the set of all candidates with paths of
length at most 2 to all other candidates
\cite{moulin:choosing-tournament};
an example of such a candidates is the candidate $x$ with maximum
outdegree, which is selected by the Copeland rule.
\cite{anshelevich:bhardwaj:elkind:postl:skowron} show that any
candidate in the uncovered set of $G$ has distortion at most 5,
and also show a lower bound of 3 on the distortion of
\emph{every} deterministic voting mechanism.

One advantage of some of the mechanisms with large distortion ---
such as Plurality, Veto, or $k$-approval with small $k$ ---
is that they require little communication from the voters.
Instead of having to transmit her entire ranking,
a voter under Plurality only needs to share her first choice;
similarly a voter under Veto only needs to share her last choice.
This observation raises the question of whether high distortion is
\emph{inherently} a consequence of limited communication between
voters and the mechanism.

The answer to the preceding question is clearly ``No:''
there are simple \emph{randomized} mechanisms achieving constant distortion.
Perhaps the simplest is Random Dictatorship:
``Return the first choice of a uniformly random voter.''
This mechanism is known to have distortion strictly smaller
than 3 \cite{anshelevich:postl:randomized},
a smaller distortion than any deterministic mechanism can achieve.
However, despite the frequent mathematical appeal and elegance of
randomized algorithms and mechanisms,
most organizations are leery of using randomization for
making important decisions;%
\footnote{A reader taking issue with this statement may want to think
  about his/her own computer science, mathematics, economics, or
  operations research department.
  Even though these are likely among the most savvy
  organizations in terms of understanding randomization, decision
  making procedures practically never involve randomization, except
  the occasional coin flip to break a tie. (And no, the fact that most
  of your colleagues seem to vote essentially randomly does not count!)

  The reasons for such a preference in most organization likely
  include an aversion to variance or low-probability undesirable
  events; naturally, one can envision guarantees between deterministic
  and expectation bounds, such as the bounds on the squared distortion
  in \cite{fain:goel:munagala:prabhu:referee}.
}
hence, we consider determinism a very desirable property in the design
of voting mechanisms.
Considering the following three properties: (1) low distortion, (2)
low communication, (3) determinism, it is known that any two can be
achieved simultaneously:
\begin{itemize}
  \item Random Dictatorship satisfies (1), (2).
  \item Uncovered-set mechanisms satisfy (1), (3).
  \item Plurality and many other mechanisms satisfy (2), (3).
\end{itemize}

The big-picture question we investigate in this article is the
tradeoff between all \emph{three} of these desirable properties.

\subsection{Our Models and Results}
We only consider the goal of minimizing the
\emph{average} (or total)
metric distance of all voters from the winning candidate.%
\footnote{Recall that \cite{anshelevich:bhardwaj:postl} and several
  follow-up articles studied both the average and median distance.}
Our main result, proved in Section~\ref{sec:general-lower},
is essentially a negative answer to the question of whether any voting
mechanisms can simultaneously have all three desirable properties.
We consider a model in which each voter communicates $b$ bits of
information about her ranking to the mechanism, in a single round.%
\footnote{Analyzing the distortion of multi-round deterministic
  mechanisms with limited communication is a very interesting
  direction for future work.}
Associated with each $b$-bit string \MSG is a subset \PSet{\MSG} of
rankings.
The \PSet{\MSG} must form a disjoint cover of all possible rankings.
If they did not form a cover, some voters might not have any message
to send, making the mechanism ill-defined.
And if the \PSet{\MSG} were not disjoint,
then it is not clear how a voter with multiple possible messages \MSG
would make the (non-deterministic) choice which one to send;
in particular, this choice could depend on the actual metric
distances, and it might require much more subtle definitions to
place meaningful restrictions on a mechanism to not exploit such
information.
Each voter communicates the (unique) \MSG such that her permutation
is in \PSet{\MSG}. 
We require that the same set \PSet{\MSG} is associated with the string \MSG, 
regardless of the identity of the voter sending the string.\footnote{%
  Our results require this assumption. While studying the power of
  mechanisms that allow different voters to use different encodings of
  their preferences would be interesting theoretically, voting
  mechanisms which treat votes differently \emph{a priori} tend to not
  be accepted in practice.}
Under this model, in Section~\ref{sec:general-lower},
we prove the following lower bound:

\begin{theorem} \label{thm:general-lower-intro}
Every one-round deterministic voting mechanism in which each voter
sends only a $b$-bit string to the mechanism has
distortion at least $\frac{2n-4}{b} - 1$.
\end{theorem}

Most mechanisms with limited communication are of a fairly specific
form: voters can communicate only their choices in a (small) set $K$ of
$k$ positions of their ranking,
typically at the top or bottom of their ballots.
(Either giving the candidate for each such position,
or specifying them as a set, as in $k$-approval.)
For such restricted mechanisms, a simpler proof
(in Section~\ref{sec:simple-lower}) gives a lower bound
that is stronger by a factor $\Theta(\log n)$:

\begin{theorem} \label{thm:simple-lower-intro}
  Any deterministic one-round social choice rule which receives,
  from each voter,
  no information about candidates outside positions $K$ in her ranking, 
  has distortion at least
  $\frac{2n-\SetCard{K}}{\SetCard{K}}$.
\end{theorem}

The proof of Theorem~\ref{thm:simple-lower-intro} is significantly
easier and cleaner than the proof of
Theorem~\ref{thm:general-lower-intro},
while still containing some of the key ideas.
Therefore, we present the proof of
Theorem~\ref{thm:simple-lower-intro} \emph{before} that of
Theorem~\ref{thm:general-lower-intro}.

Theorem~\ref{thm:simple-lower-intro} provides a generalization
of Theorem~1 of the recent work
\cite{fain:goel:munagala:prabhu:referee},
which proves linear distortion for the special case when $K$ consists of
the top $k$ positions, for constant $k$.
In fact, \cite{fain:goel:munagala:prabhu:referee} shows these lower
bounds on the \emph{expected squared distortion} of randomized
mechanisms;
this directly implies the same bounds for deterministic mechanisms.

The fact that the lower bound of Theorem~\ref{thm:simple-lower-intro}
is stronger than that of Theorem~\ref{thm:general-lower-intro}
by a factor of $\Theta(\log n)$
is discussed in more detail in Section~\ref{sec:general-lower}.
To see it most immediately,
consider the case $\SetCard{K} = k = \omega(n/\log n), k = o(n)$.
Because $k = o(n)$,
Theorem~\ref{thm:simple-lower-intro} provides a super-constant lower bound 
on the distortion.
On the other hand, communicating the positions of $k$ candidates
requires $b = \omega(n)$ bits, so the lower bound
of Theorem~\ref{thm:general-lower-intro} is vacuous.
Closing this $\Theta(\log n)$ gap is an interesting direction for
future work, discussed in Section~\ref{sec:conclusions}.

The reason we consider Theorem~\ref{thm:general-lower-intro} our main
contribution is that it helps us pinpoint the source of high
distortion.
Several recent works have shown lower bounds on the distortion of
different specific classes of social choice rules,
such as score-based rules \cite{anshelevich:bhardwaj:elkind:postl:skowron}
or the above-mentioned top-$k$ ballots \cite{fain:goel:munagala:prabhu:referee}.
Our result implies that regardless of the intricacy of the mechanism,
low communication (within the context studied here) and determinism
are enough to force high distortion.
Communication as a measure of complexity is fairly natural, as
evidenced by the mechanisms typically used in practice for large
numbers of alternatives.
Communication can also be regarded as a proxy for cognitive effort
imposed on the voters, although admittedly, the computation of a
message \MSG in a general $b$-bit bounded mechanism may still require
the voter to first determine her full ranking of all candidates.

\smallskip

The results of Theorems~\ref{thm:general-lower-intro} and
\ref{thm:simple-lower-intro} are lower bounds, raising the question of
how small one can make a mechanism's distortion when communication is
limited.
In Section~\ref{sec:upper}, we address this question, proving the
following theorem.
  
\begin{theorem} \label{thm:upper-intro}
  There is a one-round deterministic social choice rule which,
  given only each voter's top $k$ candidates (in order),
  selects a candidate with distortion at most
  $\frac{79n}{k}$.
\end{theorem}

The deterministic social choice rule of Theorem~\ref{thm:upper-intro}
is a generalization of the Copeland rule to such top-$k$ ballots.
Up to constant factors,\footnote{%
  An application of Corollary~5.3 of \cite{DistortionDuality}
  gives an upper bound of $\frac{12n}{k}$, which, however,
  is still far from matching the lower bound.}
the bounds of Theorems~\ref{thm:simple-lower-intro} and
\ref{thm:upper-intro} match.
Closing the gap between the upper and lower bound is likely difficult,
as even for $k=n$, the best-known lower bound of 3 does not match the
best current upper bound of $2+\sqrt{5} \approx 4.23$ due to
\cite{munagala:wang:improved};
whether there is a deterministic mechanism with metric
distortion 3 is a well-known open question.
Notice also that Theorem~\ref{thm:upper-intro} implies that knowing
each voter's ranking for a constant fraction of candidates is
sufficient to achieve constant distortion,
a fact that may not be a priori obvious.

As we discussed earlier, the main focus in this article is on
\emph{deterministic} mechanisms: as discussed earlier,
the \emph{Random Dictatorship} mechanism
has distortion strictly smaller than 3, achieving small distortion and
low communication simultaneously.%
\footnote{The amount by which it is smaller is of order $1/\SetCard{V}$;
  here, $\SetCard{V}$ is the the number of voters, which we consider ``large.''}
\cite{gross:anshelevich:xia:agree} prove a nearly matching lower bound:
they show that every randomized social choice rule in which each voter
only communicates her top $k < n/2$ candidates must have distortion at least
$3-\frac{2}{\Floor{n/k}}$.
However, even for $k=1$, this leaves a gap between the upper
bound of essentially $3$ for Random Dictatorship and the lower bound
of $3-\frac{2}{n}$.
Recently, \cite{fain:goel:munagala:prabhu:referee} shrunk this gap:
they proved that the Random Oligarchy mechanism
--- which samples three voters and outputs a majority of first-place
votes if it exists, and otherwise the choice 
of a random voter among the three --- achieves expected distortion
close to $3-\frac{2}{n}$, though there still remains a small gap
between the upper and lower bounds.
As an additional result, in Section~\ref{sec:randomized},
we close this remaining gap:

\begin{theorem} \label{thm:randomized-intro}
There is a simple randomized social choice rule in which each
voter only communicates her first-choice candidate,
and which achieves distortion at most $3-\frac{2}{n}$.
\end{theorem}

\subsubsection*{Nature of Latent Distances}

The optimization objective of the mechanism is expressed in terms of
latent utilities, or more specifically, distances.
A subtle question is whether voters ``know'' their utilities for (or
distances to) candidates, or --- perhaps more philosophically ---
whether these utilities/distances are ``real.''
In general, one attractive feature of the distortion framework is that
it completely obviates the need to address this question:
when a mechanism achieves low distortion, it optimizes robustly over
all possible utility/distance functions consistent with the rankings,
and the question of whether voters could actually quantify the
utilities in a meaningful way is irrelevant.

However, when we focus on the design of mechanisms with low
communication, the question should be addressed explicitly,
as the answer has a strong impact on the design space for mechanisms.
When the mechanism designer has control not only over the
\emph{aggregation} of ballots, but also over the type of information
about voter preferences that is elicited,
this opens the door to designing mechanisms in which agents explicitly
communicate numerical estimates of their utilities for some candidates;
in turn, having such information may allow a mechanism to achieve
lower distortion (as we will see in related work below).
If agents themselves cannot quantify their utilities,
then not only is communication of a ranking imposed by the class of
typically used mechanisms, but it is inherently the only information
about the utilities that agents themselves may have access to.

Which of these two assumptions
(or something between the two along a more fine-grained spectrum)
is more realistic likely depends on the envisioned application.
For example, if software agents vote on a preferred alternative in a
mostly economically motivated setting, then it is very reasonable to
assume that the agents can compute (good approximations of) their
utilities.
On the other hand, when human voters choose between political
candidates, assuming an ability to quantify a metric distance in some
abstract space of political positions is much less realistic.
Thus, we believe that for both assumptions,
there are important and natural settings in which they are justified,
motivating studies of communication-distortion tradeoffs in both types
of scenarios.

\subsection{Related Work}

Communication complexity \cite{kushilevitz:nisan} generally studies
the required communication between multiple parties wishing to jointly
compute an outcome.
Several recent works have studied the communication required
specifically for jointly computing particular economic outcomes,
or --- conversely --- to bound the effects of limited communication on
such economic outcomes.
These include work on auctions and allocations
\cite{alon:nisan:raz:weinstein:limited-interaction,assadi:auctions-interaction,blumrosen:nisan:segal:bounded-communication,blumrosen:feldman:implementation-bounded,dobzinski:nisan:oren:efficiency-interaction},
persuasion \cite{LimitedSignaling},
and general mechanism design
\cite{mookherjee:tsumagari:communication-constraints}.
While the high-level concerns are similar across different domains,
the specific approaches and techniques do not appear to carry over.

The impact of communication more specifically
on social choice rules has been explored before;
see, for instance, \cite{boutilier:rosenschein:incomplete} for an overview.
However, most of the focus in past work has been on the number of bits
that need to be communicated in order to compute the outcome of a
\emph{particular} social choice rule, rather than on proving lower
bounds arising due to limited communication when the social choice rule
is not pre-specified. 
A classic paper in this context is by Conitzer and Sandholm
\cite{conitzer:sandholm:vote-elicitation}: they study
\emph{vote elicitation} rules, i.e., protocols by which a mechanism can
interact with voters to determine the winner under a particular voting
rule while not eliciting the full ranking information.
This raises algorithmic questions about whether the information
obtained so far uniquely determines a winner as well as incentive
issues, among others, and a large amount of follow-up literature
(e.g., \cite{ding:lin:voting-partial}) has studied these issues.
Relatedly, Conitzer \cite{conitzer:eliciting-single-peaked} studies
how many comparisons need to be elicited from voters to be able to
reconstruct their complete ranking, and shows that the number is
linear (as opposed to quadratic) when preferences are single-peaked
(on the line).

Several very recent papers have explicitly considered the tradeoff between
communication and distortion in social choice,
both in deterministic and randomized settings.
Perhaps most immediately related is recent work by
Fain et al.~\cite{fain:goel:munagala:prabhu:referee}.
Their focus is on mechanisms with extremely low communication which
achieve low expected squared distortion, a measure somewhere between
expected distortion and deterministic distortion.
They prove that the Random Referee mechanism, which asks two randomly
chosen voters for their top choices, and asks a third voter to choose
between these two choices, achieves constant expected squared
distortion.
Notice that this mechanism elicits different information from
different voters.
Theorem~1 of \cite{fain:goel:munagala:prabhu:referee} shows that this
is unavoidable, in that any mechanism that only obtains top-$k$ lists
(for constant $k$), even from all voters, must have linear expected
squared distortion, implying the same result for the distortion of
deterministic mechanisms.
Our Theorem~\ref{thm:simple-lower-intro} generalizes this result for
deterministic mechanisms to non-constant $k$ and sets other than the
top $k$ positions.

Another very related piece of work is due to
Mandal et al.~\cite{mandal:procaccia:shah:woodruff},
studying the communication-distortion tradeoff in a setting where the
voters have utilities (instead of costs) for the candidates,
and these utilities are only assumed to be non-negative and normalized,
but do not need to satisfy any other properties
(such as being derived from a metric).
The other major modeling difference between our work and
\cite{mandal:procaccia:shah:woodruff} is that they assume that agents
compute their message \MSG to the mechanism directly from their
utility vector, rather than the ranking.
In particular, the mechanism can be designed to allow voters to
express the strength of their preferences, albeit in possibly coarse
form.
This allows for a choice of deterministic/randomized algorithms in two
places: (1) the voters' computation of their message, and
(2) the mechanism's aggregation of the messages into a winner.
\cite{mandal:procaccia:shah:woodruff} give upper and lower bounds for
deterministic and randomized voting rules in this setting.

The positive/algorithmic results in
\cite{mandal:procaccia:shah:woodruff} are obtained primarily by
generalizing an approach of
Benad\`{e} et al.~\cite{BNPS:distortion-rule:application},
asking voters to communicate their top few candidates as well as a
suitably rounded version of their utility for those nominated
candidates.
The bounds are improved in some parameter regimes by having the mechanism
randomly select a subset of candidates and restricting voters to choose
from this subset.

While the results of \cite{mandal:procaccia:shah:woodruff} are clearly
directly related to our work, they are not immediately comparable.
Because the utilities are not derived from metrics, the mechanisms
need to deal with much broader classes of inputs, resulting in
(generally) weaker upper bounds and stronger lower bounds.
On the other hand, the assumption that voters can explicitly quantify
their utilities --- and hence have them elicited by a mechanism ---
gives a mechanism more power than in our setting.

Another related recent piece of work is on approval-based voting,
due to Pierczy\'{n}ski and Skowron \cite{pierczynski:skowron:approval}.
While much of this work focuses on a different notion of distortion
--- analyzing the fraction of voters who \emph{approve} of the winning
candidate in the sense of being ``close enough'' ---
\cite{pierczynski:skowron:approval} also analyzes the (traditional)
distortion of approval-based voting.
Under the type of mechanism that they consider, rather than approving
a given \emph{number} of voters (as in $k$-approval),
voters approve all candidates within a given \emph{distance} of
themselves, i.e., within a ball of given radius around themselves.
This approval radius can be voter-specific or uniform across voters.
In this context, the main result of
\cite{pierczynski:skowron:approval} is to show specific constant
distortion whenever a uniform approval radius ensures that a constant
fraction of voters, bounded away from 0 and 1, have the optimum
candidate within their approval radius.\footnote{%
  In particular, when that fraction is between \quarter and \half, the
  distortion is at most 3.}
It is of course not clear how a mechanism (or the voters) could
determine such a radius.
Also note that this type of approval-based mechanism does require
voters to quantify their distances,
rather than just interact with their individual ordinal rankings.

Note that Theorem~\ref{thm:upper-intro} can be considered as somewhat
related to this result. It shows that whenever voters communicate
their top $k$ candidates, where $k$ is a constant fraction of the
number of candidates, there is a mechanism with constant distortion.
However, in contrast to the result of \cite{pierczynski:skowron:approval},
not just the identity, but also the ranking of these top $k$
candidates must be communicated; on the other hand, the theorem makes
no assumptions about whether the optimum candidate appears in any of
these top-$k$ rankings.

Low communication complexity of voter preferences is also the focus of
a recent preprint by Bentert and Skowron
\cite{bentert:skowron:few-candidates}.
They study the more ``traditional'' goal of implementing given voting
rules with low communication \cite{boutilier:rosenschein:incomplete},
but are interested in \emph{approximate} implementation of these
rules.
To make approximation meaningful, they focus on score-based rules,
which naturally assign each candidate a score (such as Borda Count,
Plurality, or MiniMax).
Then, the quality of approximation is the ratio between the score of
the winner under full information vs.~the score of the winner under
limited communication. 
They focus on mechanisms in which each voter is asked to rank a small
subset of candidates; this subset is either the voter's top $k$
candidates (a deterministic mechanism) or a random subset of $k$
candidates (a randomized mechanism).

Given that the goal in \cite{bentert:skowron:few-candidates} is the
approximate implementation of specific scoring-based voting rules
rather than achieving low distortion, the results are not directly
comparable.
However, the techniques in Section~3.2 of
\cite{bentert:skowron:few-candidates} readily yield a randomized
mechanism with distortion $5+O(\epsilon)$ and very low communication
complexity per voter when the number of voters is sufficiently large.
By asking each voter to compare a uniformly random pair of candidates 
(see also \cite{hansen:random-pairs}),
and using the majority of returned votes,
with high probability (by Chernoff and Union Bounds), 
one obtains a tournament graph in which each directed edge $(x,y)$
corresponds to at least a $\half-\epsilon$ fraction of voters
preferring $x$ over $y$.
Then, a straightforward modification of the analysis of the distortion
of uncovered set rules in 
\cite{anshelevich:bhardwaj:elkind:postl:skowron}
(or a simple application of Corollary~5.3 in \cite{DistortionDuality})
gives a distortion of $5+O(\epsilon)$.
This rule only requires each voter to compute 1 bit in total.
However, different voters are asked to answer different questions,
which is often considered undesirable.
Furthermore, the total communication complexity is $n$ bits, whereas
the Random Dictator mechanism only needs to elicit $\log_2 n$ bits
from one voter.

The recent work of Bentert and Skowron is somewhat related to earlier
work of Filmus and Oren \cite{filmus:oren:top-k-voting}:
they are also interested in the question of when top-$k$ ballots from
voters are sufficient to obtain the correct candidate.
However, \cite{filmus:oren:top-k-voting} study this question under
probabilistic models for the ballots, significantly changing the
nature of the results.

The metric-based distortion view of social choice has proved to be a
very fruitful analysis framework.
In fact, it has been extended beyond social choice to other
optimization problems in which it is natural to assume that a
mechanism only receives ordinal information;
see, e.g., \cite{anshelevich:sekar:blind,anshelevich:ordinal}.
%In addition to the results discussed previously, several additional
%directions and results are worth mentioning.
%The journal version \cite{anshelevich:bhardwaj:elkind:postl:skowron}
%also contains an expanded discussion of additional directions only
%tangentially related to the present work.

Several modeling assumptions have been proposed that yield lower
distortion than the worst-case bounds of
\cite{anshelevich:bhardwaj:elkind:postl:skowron}.
One such assumption is termed \emph{decisiveness}
\cite{anshelevich:postl:randomized,gross:anshelevich:xia:agree}:
it posits that for every voter, there is a sufficiently clear
first choice among candidates.
When the metric space is sufficiently decisive, significantly stronger
upper bounds on the distortion can be proved.
An alternative approach was proposed in
\cite{OfThePeople,BordaRepresentative}.
The authors assumed that the candidates were ``representative,''
in that they themselves were drawn i.i.d.~uniformly from the set of
voters.
Under this assumption, the authors obtained improved expected
 distortion bounds for the case of two candidates
\cite{OfThePeople}, 
and constant expected distortion for Borda count and
several other position-based scoring rules \cite{BordaRepresentative}.

As mentioned above, the gap between the upper bound of 5
(achieved, e.g., by the Copeland rule)
and the lower bound of 3 has posed an
interesting open question for several years now.
One initial conjecture of \cite{anshelevich:bhardwaj:postl}
was that the Ranked Pairs mechanism might achieve a distortion of 3.
This conjecture was disproved by \cite{goel:krishnaswamy:munagala},
who showed a lower bound of 5 on the distortion of Ranked Pairs.
Very recently, Munagala and Wang \cite{munagala:wang:improved}
have presented a (deterministic) social choice rule with
 distortion at most $2+\sqrt{5} \approx 4.23$,
which is the first piece of progress towards closing the gap.

In our and much of the preceding work on metric voting,
the focus is on distortion, while ignoring incentive compatibility.
(Recall the strong impossibility result of 
\cite{gibbard:manipulation,satterthwaite:voting}.)
The connection between strategy proofness and distortion in this type
of setting was studied in \cite{feldman:fiat:golomb}.

\section{Preliminaries} \label{sec:preliminaries}
\subsection{Voters, Candidates, and Social Choice Rules}

There are $n$ candidates, which we always denote by lowercase letters
at the end of the alphabet.
Sets of candidates are denoted by uppercase letters,
and \ALLCANDS is the set of all candidates.
The preference order (or ranking) of voter $v$ over the candidates
is a bijection $\PERM[v] : \SET{1, \ldots, n} \to \ALLCANDS$,
mapping positions $i$ to the candidate $x = \Perm[v]{i}$ which voter
$v$ ranks in position $i$.
We say that $v$ (strictly) \emph{prefers} $x$ to $y$ iff
$\PermInv[v]{x} < \PermInv[v]{y}$.
When only the ranking, but not the identity, of a voter is relevant,
we will omit the subscript $v$ for legibility.
The set of all voters\footnote{%
  We will not need to reference the number of voters explicitly.
  In general, we treat the number of voters as ``much larger'' than
  the number of candidates, and are only interested in bounds in terms
  of the number of candidates.}
is denoted by \ALLVOTERS.
We write \ALLPERMS for the set of all possible rankings
$\PERM: \SET{1, \ldots, n} \to \ALLCANDS$,
and $\PERMMAT = (\Perm[v]{i})_{v \in \ALLVOTERS, i \in \SET{1, \ldots, n}}$
for the rankings of all voters, which we call the \emph{vote profile}.
  
In the traditional full-information view,
a \emph{social choice rule} (we use the terms
\emph{mechanism} or \emph{voting mechanism} interchangeably)
$\SCRULE : \ALLPERMS^{\ALLVOTERS} \to \ALLCANDS$
is given the rankings of all voters, i.e., \PERMMAT,
and produces as output one \emph{winning} candidate
$\WINNER = \SCRule{\PERMMAT}$.
For most of this article, we are interested only in
\emph{deterministic} social choice rules \SCRULE.

\subsection{Communication-bounded mechanisms}

Our main contribution is to consider communication-bounded social
choice rules.
As in the standard model described above,
we still only consider deterministic \emph{single-round} mechanisms,
i.e., each voter can only send a single message to the mechanism.
However, this message is now also restricted to be at most $b$ bits
long.

This induces $\PARTS = 2^b$ sets
$\PSet{1}, \PSet{2}, \ldots, \PSet{\PARTS}$ of rankings;
when the mechanism receives a message \MSG from voter $v$,
all it learns is that $\PERM[v] \in \PSet{\MSG}$.
As discussed in the introduction,
we assume that the \PSet{\MSG} form a disjoint partition of \ALLPERMS,
i.e., they are pairwise disjoint and cover all rankings:
$\bigcup_{\MSG=1}^{\PARTS} \PSet{\PARTS} = \ALLPERMS$.
%If the sets did not cover all rankings,
%then a voter with an uncovered ranking would not have a
%suitable message to send to the mechanism.
The fact that \PARTS is a power of 2 is not relevant anywhere in our
proofs,
so we also consider mechanisms with arbitrary numbers \PARTS of sets.

\begin{definition}[\PARTS-communication bounded social choice rule]
\label{def:communication-bounded-rule}
An \emph{\PARTS-communication bounded social choice rule} consists of
pairwise disjoint sets
$\PSet{1}, \PSet{2}, \ldots, \PSet{\PARTS} \subseteq \ALLPERMS$
with $\bigcup_{\MSG=1}^{\PARTS} \PSet{\PARTS} = \ALLPERMS$,
and a deterministic mapping 
$\SCRULE: \SET{1, \ldots, \PARTS}^{\ALLVOTERS} \to \ALLCANDS$.
\end{definition}

Communication-bounded social choice rules that are used in
practice, such as Plurality, Veto, $k$-approval, and combinations
thereof, are of a specific form:
there is a set $K$ of $k$ positions, and voters can communicate the
set of candidates they have in positions in $K$,
possibly with an ordering, but cannot communicate any additional
information about their ranking of candidates in positions outside $K$.
For such mechanisms, we will be able to prove stronger lower bounds on
the distortion, and with a significantly simpler proof.
We define them formally as follows:

\begin{definition} \label{def:k-entry-rule}
  A \emph{$k$-entry social choice rule} is an
  \PARTS-communication bounded social choice rule with the following
  additional restriction on the sets 
  $\PSet{1}, \PSet{2}, \ldots, \PSet{\PARTS}$:
  there exists a set $K \subseteq \SET{1, \ldots, n}$ of at most $k$
  positions such that 
  if $\PERM, \PERMP$ agree for all positions in $K$,
  i.e., $\Perm{i} = \PermP{i}$ for all $i \in K$, then
  $\PERM \in \PSet{\MSG}$ if and only if $\PERMP \in \PSet{\MSG}$.
\end{definition}

\subsection{Metric Space and Distortion}

The key modeling contribution of the metric-based distortion
\cite{anshelevich:bhardwaj:postl}
objective is to assume that all voters and candidates are 
embedded in a pseudo-metric space \DIST.
\Dist{v}{x} denotes the distance between voter $v$ and candidate $x$.
Being a pseudo-metric, it satisfies non-negativity and the triangle
inequality $\Dist{v}{x} \leq \Dist{v}{y} + \Dist{v'}{y} + \Dist{v'}{x}$
for all voters $v, v'$ and candidates $x, y$.
Given our choice of defining the metric only for pairs consisting of a
voter and a candidate, symmetry is not directly relevant.
One can naturally extend the pseudo-metric to pairs of candidates or
pairs of voters, but those distances will never appear in our
mechanisms or proofs.
For our upper bounds, we explicitly allow the distance between
candidates and voters
(and thus also between pairs of candidates or pairs of voters)
to be 0; however, for improved flow, we will still refer to \DIST as
a \emph{metric}.
In our lower-bound constructions, all distances will be strictly
positive; that is, we do not exploit the increased generality for
negative results.

We say that a vote profile \PERMMAT is \emph{consistent} with the
metric \DIST, and write \Legal{\DIST}{\PERMMAT},
if $\Perm[v]{x} < \Perm[v]{y}$ whenever $\Dist{v}{x} < \Dist{v}{y}$.
That is, \PERMMAT is consistent with \DIST iff all
voters rank candidates by non-decreasing distance from themselves.
Notice that in case of ties among distances,
i.e., $\Dist{v}{x} = \Dist{v}{y}$,
several vote profiles are consistent with \DIST.
None of our results depend on any tie breaking assumptions.

The \emph{cost} of candidate $x$ is
$\Cost{x} = \sum_{v} \Dist{v}{x}$,
i.e., the sum of distances of $x$ to all voters.%
\footnote{\cite{anshelevich:bhardwaj:postl} also consider the median
  distance as an optimization objective;
  here, we only focus on the sum/average objective.}
An \emph{optimum} candidate is any candidate
$\OPT[\DIST] \in \argmin_{x \in \ALLCANDS} \Cost{x}$;
in our analysis, it will not matter which candidate is considered
``the'' optimum candidate in case of ties.

The social choice rule is handicapped by not knowing the metric \DIST,
instead only observing the consistent vote profile \PERMMAT
(or some limited information about it, when communication is restricted).
Due to this handicap, and possibly other suboptimal choices,
it will typically choose candidates with higher cost than
\Cost{\OPT}.
The \emph{distortion} of \SCRULE is the worst-case ratio between the
cost of the candidate chosen by \SCRULE,
and the optimal candidate \OPT[\DIST]
(determined with knowledge of the actual distances \DIST).
Formally,
\[
  \Distortion{\SCRULE} \; = \;
  \max_{\PERMMAT} \sup_{\DIST: \Legal{\DIST}{\PERMMAT}}
  \frac{\Cost{\SCRule{\PERMMAT}}}{\Cost{\OPT[\DIST]}}.
\]
We can think of the distortion in terms of a game between the social
choice rule and an adversary.
First, the adversary chooses the vote profile \PERMMAT.
Then, the social choice rule, knowing only \PERMMAT
(or part of that information, in case of communication restrictions),
chooses a winning candidate \WINNER = \SCRule{\PERMMAT}.
Then, the adversary chooses a metric \DIST consistent with
\PERMMAT that maximizes the ratio between the cost of the
candidate chosen by \SCRULE and the optimum candidate for \DIST.

The goal now is to define a social choice rule \SCRULE
--- under suitable constraints --- that achieves small distortion
\Distortion{\SCRULE}, and to prove lower bounds on all social choice
rules under the given constraints.

\section{A Lower Bound for $k$-Entry Social Choice Rules} \label{sec:simple-lower}
In this section, we establish the lower bound of
Theorem~\ref{thm:simple-lower-intro}, restated here formally.

\begin{theorem} \label{thm:simple-lower}
  Every one-round deterministic $k$-entry social choice rule has
  distortion at least $\frac{2n-k}{k}$.
\end{theorem}

\begin{emptyproof}
Let $K = \SET{\kappa_1 < \kappa_2 < \cdots < \kappa_k}$.
Because \emph{every} deterministic social choice rule has distortion
at least 3 \cite{anshelevich:bhardwaj:postl},
we only need to consider the case where $2n-k > 3k$, i.e., $k < n/2$.
We will prove the theorem by induction on $n$,
with the base case $n = 2$ holding because the only such case with $k
< n/2$ is $k=0$, where the mechanism receives no information about any
voter's preferences, and hence has unbounded distortion.

First, we consider the case when $n \in K$.
We designate one candidate $\hat{x}$ who is
``infinitely'' far from all other candidates and voters,
and thus ranked last by all voters.
The mechanism clearly cannot choose $\hat{x}$ as a winner.
This reduces the problem to one of $n-1$ candidates,
and a set $K' = K \setminus \SET{n}$ of $k-1$ positions at which
voters specify their ranking.
By induction hypothesis, applied to this instance,
the distortion is lower-bounded by
$\frac{2(n-1)-(k-1)}{k-1} = \frac{2n-k-1}{k-1} > \frac{2n-k}{k}$;
the inequality holds because $k < n$.

For the remainder of the proof, we can assume that $n \notin K$,
i.e., voters do not specify their least favorite candidate.
In this case, we will not need to use the induction hypothesis for $n-1$.
For each subset $S \subseteq \ALLCANDS, \SetCard{S} = k$ of $k$ candidates,
and each ordering $\sigma : \SET{1, \ldots, k} \to S$,
we say that a voter $v$ has \emph{type} $(S, \sigma)$
if she puts the candidates from $S$ in the positions $K$,
in the order given by $\sigma$.
That is, $v$ has type $(S, \sigma)$ iff
$\Perm[v]{\kappa_i} = \sigma(i)$ for $i = 1, \ldots, k$.
There are $t = {n \choose k} \cdot k!$ types of voters.
We define a vote profile which has exactly a $1/t$ fraction
of voters of type $(S, \sigma)$, for each type.
Throughout, we will talk about fractions, rather than numbers, of
voters, so that the total adds up to 1.

Each subset of candidates
and each order among those candidates is equally frequent,
and in aggregate, the vote profile expresses no preference
by the voters for any candidate over any other.
%\footnote{%
%  If one wanted to avoid having to
%  break such ``perfect ties,'' one could perturb these fractions by
%  some sufficiently small $\delta$, which would only affect the
%  statement by a suitable function of $\delta$, which goes to 0 as
%  $\delta \to 0$. The statement of the theorem would still be obtained.}
%Let \Proj{K} be defined accordingly.
Let \WINNER be the candidate chosen by the social choice rule for
this input.
\WINNER is well-defined as a function of all voters' types,
because
(1) for each voter $v$, the message sent by $v$ is uniquely determined
by her ranking of candidates in positions in $K$, and
(2) the mechanism's output is a deterministic function of only the
messages sent by the voters.

We now define a metric space.
Let $\epsilon$ be a very small constant (we will let $\epsilon \to 0$),
and $0 < \epsilon_1 < \epsilon_2 < \cdots < \epsilon_{n} < \epsilon$.
Consider a voter $v$ of type $(S, \sigma)$.
We distinguish two cases:
\begin{enumerate}
  \item In the first case, $\WINNER \notin S$.
      Let \PERM[v] be any ordering that puts the candidates in $S$
      in positions $K$ in the order $\sigma$,
      and which additionally has $\Perm[v]{n} = \WINNER$,
      i.e., candidate \WINNER is in the last position in $v$'s ranking.
      Apart from this, \PERM[v] is arbitrary.
      By construction, a voter $v$ with ranking \PERM[v] has type
      $(S, \sigma)$.
      We now set the distance between $v$ and the candidate \WINNER
      to 1,
      and the distance from $v$ to every candidate \Perm[v]{i}
      (for $i < n$) to $\epsilon + \epsilon_i$.
      These distances are consistent with the ranking \PERM[v].
  \item In the second case, $\WINNER \in S$.
    Again, let \PERM[v] be any permutation that puts the candidates in
    $S$ in positions $K$ in the order $\sigma$ (ensuring that \PERM[v]
    is consistent with $v$ having type $(S, \sigma)$).
    This time, the position of \WINNER in \PERM is prescribed by
    $S, \sigma$, and we let the remaining positions of \PERM[v] be
    arbitrary.
    Voter $v$ has distance exactly $\half + \epsilon + \epsilon_i$
    from each candidate \Perm[v]{i},
    including the case when $\Perm[v]{\WINNER} = i$.
    Again, $v$ ranks the candidates in the order given by \PERM[v].
\end{enumerate}
%\end{description}
We now verify that these distances satisfy the triangle inequality.
Consider voters $v, v'$ and candidates $x, y$.
We will show that $\Dist{v}{y} \leq \Dist{v}{x} + \Dist{v'}{x} +
\Dist{v'}{y}$, by distinguishing two cases for $y$:
\begin{enumerate}
\item In the first case, $y = \WINNER$. 
Then, $\half + \epsilon \leq \Dist{v}{y} \leq 1$.
Either the distance $\Dist{v'}{y} = 1$,
in which case the triangle inequality holds obviously,
or $\Dist{v'}{y} \geq \half + \epsilon$,
in which case our definition ensures that
$\Dist{v'}{x} \geq \half + \epsilon$ as well.
In either case, the triangle inequality holds.

\item In the second case, $y \neq \WINNER$,
so either $\epsilon < \Dist{v}{y} < 2 \epsilon$ or
$\half + \epsilon < \Dist{v}{y} < \half + 2 \epsilon$,
depending on the case of the definition.
Because all distances are lower-bounded by $\epsilon$,
the triangle inequality clearly holds if $\Dist{v}{y} < 2 \epsilon$.
In the other case $\half + \epsilon < \Dist{v}{y}$, we have that
$\half + \epsilon < \Dist{v}{x}$, which together with
$\epsilon < \Dist{v'}{x}$ again ensures that the triangle inequality
holds.
\end{enumerate}

Recall that \WINNER is selected by 
the social choice rule under the given rankings.
Each voter of type $(S, \sigma)$ with $\WINNER \notin S$
has cost 1 for candidate \WINNER,
and cost at most $2 \epsilon$ for any candidate $x \neq \WINNER$.
Each voter of type $(S, \sigma)$ with $\WINNER \in S$
has cost at least \half for candidate \WINNER, 
and cost at most $\half + 2\epsilon$ for each candidate $x \neq \WINNER$.

Of the $t$ types $(S, \sigma)$, exactly ${n-1 \choose k-1} \cdot k!$ have
$\WINNER \in S$.
Thus, the cost of candidate \WINNER is at least
$\frac{1}{t} \cdot (\half \cdot {n-1 \choose k-1} \cdot k!
+ 1 \cdot (t - {n-1 \choose k-1} \cdot k!))$,
while the cost of any other candidate is at most
$\frac{1}{t} \cdot (2\epsilon + \half \cdot {n-1 \choose k-1} \cdot k!)$.
Letting $\epsilon \to 0$, the distortion approaches
\[
  1 + \frac{2 (\frac{n!}{(n-k)!} - \frac{k \cdot (n-1)!}{(n-k)!})}%
           {\frac{k \cdot (n-1)!}{(n-k)!}}
  \; = \; 1 + \frac{2(n-k)}{k}
  \; = \; \frac{2n-k}{k}. \QED
\]
\end{emptyproof}

\section{The General Lower Bound} \label{sec:general-lower}
In this section, we prove the more general lower bound
of Theorem~\ref{thm:general-lower-intro}.
The bound applies to all \PARTS-communication bounded social choice rules, 
but is slightly weaker than that of Theorem~\ref{thm:simple-lower}.
To gain some insight into general communication-bounded social choice rules,
we begin with an easy proposition, independently obtained as Lemma~4.1
in \cite{mandal:procaccia:shah:woodruff}.
We include a proof here for completeness, and because it illustrates
some of the type of reasoning required for the proof of
Theorem~\ref{thm:general-lower-intro}.

\begin{proposition} \label{prop:first-element}
  Assume that there exists a set \PSet{\MSG} containing two rankings
  \PERM, \PERMP with $\Perm{1} \neq \PermP{1}$,
  i.e., there is a \MSG which does not uniquely specify the voter's
  top-ranked candidate.
  Then, the corresponding social choice rule has unbounded distortion.
\end{proposition}

\begin{proof}
  Let $x = \Perm{1}, y = \PermP{1}$.
  Consider a vote profile in which all voters communicate the message
  \MSG to the mechanism, i.e., state that their ranking is in \PSet{\MSG}.
  If the mechanism chooses $x$ as the winner, then the metric will
  be such that all voters have distance 0 from $y$,
  and distance 1 from all other candidates\footnote{%
    At the cost of small $\epsilon_i$, which we could then let go to 0,
    we could avoid ties here; in the limit, we would obtain exactly
    the same result. See the proof of Theorem~\ref{thm:simple-lower}
    for spelled-out details.}, including $x$.
  Then, the cost of $y$ is 0, while the cost of $x$ is 1,
  giving infinite cost ratio, i.e., distortion.
  Similarly, if the mechanism does not choose $x$ as the winner,
  then all voters will be at distance 0 from $x$ and at distance 1
  from all other candidates, including $y$.
  Again, the cost ratio between the optimum candidate $x$ and the
  winner will be infinite.
\end{proof}

%\begin{remark}
%  Proposition~\ref{prop:first-element} and its proof are practically
%  identical to Lemma~4.1 from \cite{mandal:procaccia:shah:woodruff}.
%  As discussed there as well, it has a few immediate implications.
%  As discussed there as well,
%  it implies that if a mechanism restricts communication to
%  fewer than $\log_2 n$ bits, it has unbounded distortion,
%  because the voters cannot even specify
%  their first-choice candidate unambiguously.
%  Second, this result directly generalizes Example~8 of
%  \cite{anshelevich:bhardwaj:elkind:postl:skowron},
%  which showed unbounded distortion of positional scoring rules which
%  give the same score to the first two candidates
%  (which includes the Veto and $k$-approval rules).
%  Proposition~\ref{prop:first-element} perhaps sheds some more general
%  light on the reason: any voting rule that does not completely
%  disambiguate voters' first choices must have distortion that cannot be
%  bounded in terms of the number of candidates $n$.
%\end{remark}

\begin{theorem} \label{thm:general-lower}
Let \SCRULE be any one-round \PARTS-communication bounded social
choice rule on $n$ candidates.
Then, \SCRULE must have distortion at least
$\frac{2n-4}{\ln \PARTS} - 1$.
\end{theorem}

\begin{proof}
  The high-level idea of the proof is to use induction on the number
  of candidates,
  to show that when communication is ``sufficiently bounded,''
  any social choice rule must have high distortion.
  After completing the proof by induction, we would like to apply the
  result to $n$ candidates, and ``sufficiently bounded'' must then
  include \PARTS-communication bounded.
  Therefore, the relationship between the number of candidates in the
  induction proof and the bound on communication depends on
  $n, \PARTS$, and to avoid notational ambiguity,
  we will use different variable names for the induction.
  Specifically, we use \IC for the number of candidates within the
  induction proof, and \PARTS[\IC] for the upper bound on
  communication.

  Let $\gamma = 1-\PARTS^{-1/(n-2)}$.
  We will prove by induction on \IC that every
  \PARTS[\IC]-communication bounded social choice rule on \IC candidates
  with $\PARTS[\IC] \leq \frac{1}{(1-\gamma)^{\IC-2}}$ has distortion at least
  $\frac{2}{\gamma} - 1$.

  The base case $\IC=2$ is easy:
  the communication bound is $\PARTS[2] \leq \frac{1}{(1-\gamma)^{2-2}} = 1$,
  so the voters cannot communicate any preference.
  By Proposition~\ref{prop:first-element},
  the social choice rule has unbounded distortion.
  For the induction step, we distinguish two cases:
  \begin{enumerate}
  \item In the first case, we assume that for each candidate $x$,
    at least a $1-\gamma$ fraction of all sets \PSet{\MSG} contain
    a ranking $\PERM[\MSG] \in \PSet{\MSG}$ that ranks $x$ last,
    i.e., $\Perm[\MSG]{\IC} = x$.
    Then, we consider a vote profile with \PARTS[\IC] voters in which for
    each $\MSG = 1, \ldots, \PARTS[\IC]$, exactly one voter submits \MSG.

    Let \WINNER be the candidate chosen by \SCRULE.
    Consider the following metric space: 
    For every voter $v$ who submitted \MSG such that there is a ranking
    $\PERM[\MSG] \in \PSet{\MSG}$ ranking \WINNER last,
    we define the distance between $v$ and \WINNER to be 1,
    and the distance from all other candidates\footnote{%
      Again, ties could be broken by using
      small $\epsilon_i \to 0$ without affecting the final result.}
    to be 0.
    For all other voters, the distance to all candidates is \half.
    Said differently, all candidates $x \neq \WINNER$ are at distance
    0 from each other, and at distance 1 from \WINNER.
    All voters who could possibly rank \WINNER last are in the same
    location as the candidates different from \WINNER,
    while all other voters are halfway between \WINNER and the other
    candidates.

    Then, the cost of \WINNER is at least
    $\gamma \cdot \half + (1-\gamma) \cdot 1 = 1-\frac{\gamma}{2}$,
    while the cost of each other candidate is at most
    $\gamma \cdot \half + (1-\gamma) \cdot 0 = \frac{\gamma}{2}$.
    Thus, the distortion of the mechanism is at least
    $\frac{2}{\gamma} - 1$, completing the proof directly.
    
  \item Otherwise, let $x$ be a candidate such that at most a
    $1-\gamma$ fraction of all sets \PSet{\MSG} contain a ranking
    $\PERM[\MSG] \in \PSet{\MSG}$ that ranks $x$ last.
    Define \PARTS[\IC-1] to be the number of such sets,
    and assume w.l.o.g.~(by renumbering) that
    $\PSet{1}, \PSet{2}, \ldots, \PSet{\PARTS[\IC-1]}$ are all the sets
    which contain at least one ranking with $x$ in the last position.
    By the assumption in this part of the proof, we have that
    $\PARTS[\IC-1] \leq \Floor{(1-\gamma) \cdot \PARTS[\IC]}$.
    We will only construct instances in which all voters rank $x$ last;
    thus, no voter communicates any message $\MSG > \PARTS[\IC-1]$.
%    and we will typically also have ruled out some rankings that were
%    in the sets \PSet{\MSG} for $\MSG \leq \PARTS[\IC-1]$.

    No mechanism with finite distortion can select $x$ as a
    winner, by the same argument as in the preceding case.
    (That is, the metric puts $x$ at distance 1 from all voters,
    and all other candidates at distance 0 from all voters.)
    As a result, we obtain an instance with $\IC-1$ candidates,
    only $(\IC-1)!$ remaining possible rankings,
    and --- crucially --- only
    $\PARTS[\IC-1] \leq (1-\gamma) \cdot \PARTS[\IC]$
    remaining sets of rankings.
    We can therefore apply the induction hypothesis for $\IC-1$,
    and conclude that the mechanism's distortion is at least
    $\frac{2}{\gamma} - 1$.
  \end{enumerate}

  To show that we can apply the inductive claim with $\IC=n$ in the end,
  observe that
  $\PARTS[n] = \PARTS = \PARTS^{(n-2)/(n-2)} = \frac{1}{(1-\gamma)^{n-2}}$.
  
  It remains to show that
  $\frac{2}{\gamma} - 1 \geq \frac{2n-4}{\ln \PARTS} - 1$.
  To do so, we rewrite $\gamma$ by using the Taylor expansion of
  $t^{1/(n-2)}$ around $t=1$, then apply straightforward bounds:
  \begin{align*}
    \gamma & = 1-\PARTS^{-1/(n-2)} \\
    & = \frac{1}{n-2} \sum_{k=1}^{\infty} \frac{1}{k} \cdot (1-1/\PARTS)^k
      \cdot \prod_{j=1}^{k-1} \left(1-\frac{1}{j \cdot (n-2)}\right)\\
    & \leq \frac{1}{n-2} \sum_{k=1}^{\infty} \frac{1}{k} \cdot (1-1/\PARTS)^k\\
    & = \frac{1}{n-2} \cdot \ln \PARTS.
  \end{align*}
  Substituting this bound for $\gamma$ into the distortion completes
  the proof.
\end{proof}

To compare the bound of Theorem~\ref{thm:general-lower} with that 
of Theorem~\ref{thm:simple-lower},
observe that when voters get to specify the candidates in each of $k$
(given) positions in a ranking,
this generates a partition of \ALLPERMS into
$\PARTS = {n \choose k} \cdot k! = \frac{n!}{(n-k)!}$ sets:
one for each subset and order within that subset.
These sets of rankings do in fact form a disjoint cover.
For the ``interesting'' range $k \leq n/2$, we can simply bound
$(n/2)^k \leq \PARTS \leq n^k$,
so we get that $\ln \PARTS \approx k \ln n$.
This shows that the lower bound of Theorem~\ref{thm:general-lower} is
weaker than that of Theorem~\ref{thm:simple-lower} by a factor of
$\Theta(\log n)$.
Closing this gap is an interesting direction for future work,
briefly discussed in Section~\ref{sec:conclusions}.

\section{A Near-Matching Upper Bound} \label{sec:upper}
While the results of Theorems~\ref{thm:simple-lower}
and \ref{thm:general-lower} are negative,
there are parameter ranges, such as $k = o(n), k = \omega(1)$,
in which they leave room for non-trivial positive results,
in particular, sublinear distortion.
In this section, we investigate how well one-round mechanisms can do
with limited communication.

Our main result is a $k$-entry social choice rule which
--- up to constants ---
matches the lower bound of Theorem~\ref{thm:simple-lower}.
This shows that the lower bound of Theorem~\ref{thm:simple-lower} is
essentially tight.
Not surprisingly, the mechanism is a variation on
uncovered set mechanisms, which are the only type of mechanism known
to achieve constant distortion even with access to the full vote profile.

In our mechanism, each voter communicates her top $k$ choices.
We say that voter $v$ \emph{prefers} $x$ over $y$ if either:
(1) Both $x$ and $y$ are among her top choices, and she ranks $x$
higher than $y$, or
(2) $x$ is among her top choices, and $y$ is not.
Obviously, the mechanism does not know which of two candidates she
prefers if neither candidate is among her top $k$ candidates.

As in uncovered set mechanisms like Copeland, we construct a
\emph{comparison graph} $G$ among the $n$ candidates.
Define $\alpha = \frac{k}{3n}$.
For each ordered pair $x, y$, the graph $G$ contains a directed edge
$(x,y)$ if and only if at least an $\alpha$ fraction of all voters
prefer $x$ over $y$.
Notice that because $\alpha \leq \half$, it is possible that $G$ contains
both $(x,y)$ and $(y,x)$.
Similarly, it is possible that for a pair $\SET{x,y}$,
$G$ contains neither $(x,y)$ nor $(y,x)$;
for instance, this will happen if no voter ranks either $x$ or $y$
among her top $k$ candidates.

Let $S_2$ be the set of candidates $x$ such that at least a $2\alpha$
fraction of voters rank $x$ among their top $k$ candidates.
(We will show in the proof of Lemma~\ref{lem:short-paths} that $S_2$ is not empty.)
The winner \WINNER returned by \MECH is a
candidate in the induced graph $G[S_2]$ with largest outdegree; 
notice that edges leaving $S_2$ are \emph{not} counted.

\begin{theorem} \label{thm:upper}
  \MECH has distortion at most $\frac{79n}{k}$.
\end{theorem}

% \begin{remark}
%   Up to constants, the upper bound of Theorem~\ref{thm:upper}
%   matches the lower bound from Theorem~\ref{thm:simple-lower}.
%   The fact that there is some constant-factor gap between the upper
%   and lower bounds is not surprising.
%   Indeed, even for the case when the mechanism sees the full
%   ranking of each voter, the best known upper bound is 5, while
%   the best lower bound is 3 \cite{anshelevich:bhardwaj:postl}.
%   In our proof, we have made no attempt to optimize the constants;
%   as discussed in Remark~\ref{rem:lemma-not-tight},
%   the upper bound can be improved by a factor of more than 6.
% \end{remark}

We begin with a lemma showing the key structural property of the
winning candidate \WINNER.% within the comparison graph $G$.

\begin{lemma} \label{lem:short-paths}
  In $G$, for every candidate $x$, there is a directed path of length
  at most 3 from \WINNER to $x$.
\end{lemma}

\begin{proof}
  Similar to the definition of $S_2$,
  let $S_3$ be the set of candidates $x$ such that at least a
  $3\alpha$ fraction of the voters ranks $x$ somewhere among their top
  $k$ candidates.
  By the Pigeon Hole Principle, because each voter ranks a
  $\frac{k}{n}$ fraction of candidates in her top $k$,
  and $\alpha = \frac{k}{3n}$, at least one candidate occurs in
  a $3\alpha$ fraction of top-$k$ lists.
  In particular, $S_3$ (and thus $S_2$) is non-empty.

  Each candidate $x \in S_3$ has a directed edge to each
  candidate $y \notin S_2$.\footnote{%
  The opposite edge may be in $G$ as well; this is irrelevant.}
  This is because $x$ appears in at least a $3\alpha$ fraction of
  top-$k$ lists, while $y$ appears in at most a $2\alpha$ fraction.
  In particular, at least an $\alpha$ fraction of voters rank $x$,
  but not $y$, in their top-$k$ lists, and thus prefer $x$ to $y$.

  Now consider the induced graph $G[S_2]$.
  For each pair $x, y \in S_2$, at least one of the edges $(x,y)$ or
  $(y,x)$ is in $G[S_2]$.
  This is because of the (at least) $2\alpha$ fraction of voters with
  $x$ in their lists, at least an $\alpha$ fraction rank $y$ higher in
  their lists,
  or at least an $\alpha$ fraction rank $y$ lower
  (or not in their lists).
  Hence, $G[S_2]$ is a supergraph of a tournament graph.

  Because \WINNER has maximum degree in $G[S_2]$, it also has maximum
  degree in at least one tournament subgraph of $G[S_2]$.
  It is well known
  (see, e.g., \cite{moulin:choosing-tournament,anshelevich:bhardwaj:postl})
  that the maximum-degree node in a tournament graph 
  is in the uncovered set,
  i.e., it has a directed path of length at most 2 to every other node.
  This of course still holds in the supergraphs $G[S_2]$ and $G$.
  Thus, \WINNER has a directed path of length 2 in $G$
  to every candidate $x \in S_2$.

  Let $y \in S_3$ be arbitrary.
  By the preceding two paragraphs,
  $y$ has a directed edge to each $x \notin S_2$,
  and \WINNER has a directed path of length at most 2 to $y$.
  In summary, \WINNER has a directed path of length at most 3 to each
  candidate $x$.
\end{proof}

Next, we show a lemma upper-bounding the cost ratio of two candidates
$x, y$ when $x$ has a directed path of length at most 3 to $y$.

\begin{lemma} \label{lem:cost-ratio}
Let $\WINNER, z$ be two candidates such that there is a directed path of
length at most $\ell$ edges from \WINNER to $z$ in $G$.
Then, $\Cost{\WINNER} \leq (1+\frac{3^{\ell}-1}{\alpha}) \cdot \Cost{z}$.
\end{lemma}

\begin{remark} \label{rem:lemma-not-tight}
Lemma~\ref{lem:cost-ratio} can be considered a (somewhat weaker)
generalization of a result proved in the proof of Theorem~7 in
\cite{anshelevich:bhardwaj:postl}
(see also the discussion in the subsequent remark in
\cite{anshelevich:bhardwaj:postl}).
By Lemma~6 of \cite{anshelevich:bhardwaj:postl},
if an $\alpha$ fraction of voters prefer $x$ over $y$,
then $\Cost{x} \leq 1 + \frac{2(1-\alpha)}{\alpha} \cdot \Cost{y}$.
In particular, for $\alpha = \half$, this implies an upper bound of 3
on the cost ratio.
If $x$ has a directed path of length $\ell$ to $y$, then this bound
implies\footnote{%
Theorem~7 of \cite{anshelevich:bhardwaj:postl} uses a more intricate
proof to improve the upper bound for length-2 paths from this
immediate 9 to 5.}
that $\Cost{x} \leq 3^\ell \cdot \Cost{y}$.
However, since we are interested in a regime where $\alpha = o(1)$,
the exponential dependence on $\ell$ (recall that we have $\ell=3)$
would result in bounds that do not match our lower bounds asymptotically.
The point of Lemma~\ref{lem:cost-ratio} is to improve upon this
exponential dependence.

The exponential dependence on $\ell$ is an artifact of our relatively
simple proof.
Applying Corollary~5.3 from \cite{DistortionDuality} instead
would yield an improved bound of
$\frac{\ell}{\alpha}+1$ or $\frac{\ell+1}{\alpha}-1$,
depending on whether $\ell$ is even or odd.
\end{remark}

\begin{proof}
  Let $(\WINNER,y_1,y_2, \ldots, y_{\ell-1},y_{\ell} := z)$
  be a directed path of $\ell$ edges from \WINNER to $z$.
  We distinguish two cases, based on the relative lengths of the
  distances \Dist{\WINNER}{y_i} and \Dist{y_i}{z},
  compared to \Dist{\WINNER}{z}.

  \begin{enumerate}
  \item If there exists a candidate $y_i$ (with $i < \ell$) such that
    $\Dist{y_i}{y_{i+1}} \geq \frac{2 \cdot 3^{\ell-i-1}}{3^{\ell}-1}
    \cdot \Dist{\WINNER}{z}$,
    then let $i$ be maximal with this property.    

    All the voters who prefer $y_i$ over $y_{i+1}$,
    which comprise at least an $\alpha$ fraction of all voters,
    are at distance at least
    $\frac{\Dist{y_i}{y_{i+1}}}{2} \geq \frac{3^{\ell-i-1}}{3^{\ell}-1}
    \cdot \Dist{\WINNER}{z}$ from $y_{i+1}$.

    By maximality of $i$, all candidates $y_j$ with $j > i$ have
    $\Dist{y_j}{y_{j+1}} < \frac{2 \cdot 3^{\ell-j-1}}{3^{\ell}-1}
    \cdot \Dist{\WINNER}{z}$.
    Using the triangle inequality and summing this inequality for all
    $j > i$ gives us that
    $\Dist{y_{i+1}}{z} < \frac{2 \cdot 3^{\ell-1}}{3^{\ell}-1}
    \cdot \Dist{\WINNER}{z} \cdot \sum_{j=i+1}^{\ell-1} 3^{-j}
    = \frac{2 \cdot 3^{\ell-1}}{3^{\ell}-1}
    \cdot \Dist{\WINNER}{z} \cdot \frac{3^{-i)}-3^{-(\ell-1)}}{2}
    = \frac{3^{\ell-1-i}-1}{3^{\ell}-1} \cdot \Dist{\WINNER}{z}$.

    Again by triangle inequality, the voters who prefer $y_i$ over
    $y_{i+1}$ are at distance at least
    $\Dist{y_i}{z} \geq \frac{\Dist{y_i}{y_{i+1}}}{2} - \Dist{y_{i+1}}{z}
    \geq \frac{3^{\ell-i-1}}{3^{\ell}-1} \cdot \Dist{\WINNER}{z}
       - \frac{3^{\ell-1-i}-1}{3^{\ell}-1} \cdot \Dist{\WINNER}{z}
       = \frac{1}{3^{\ell}-1} \cdot \Dist{\WINNER}{z}$
    from $z$.
    
  \item In the other case, all candidates $y_i$ with $i < \ell$ have
    $\Dist{y_i}{y_{i+1}} < \frac{2 \cdot 3^{\ell-i-1}}{3^{\ell}-1}
    \cdot \Dist{\WINNER}{z}$.
    Again, using the triangle inequality and summing this inequality
    for all $i$, we can bound
    $\Dist{y_{1}}{z} < \frac{2 \cdot 3^{\ell-1}}{3^{\ell}-1}
    \cdot \Dist{\WINNER}{z} \cdot \sum_{j=1}^{\ell-1} 3^{-j}
    = \frac{2 \cdot 3^{\ell-1}}{3^{\ell}-1}
    \cdot \Dist{\WINNER}{z} \cdot \frac{1-3^{-(\ell-1)}}{2}
    = \frac{3^{\ell-1}-1}{3^{\ell}-1} \cdot \Dist{\WINNER}{z}$.

    Therefore, by triangle inequality,
    $\Dist{\WINNER}{y_1} \geq \Dist{\WINNER}{z} - \Dist{y_1}{z}
    > \frac{3^{\ell} - 3^{\ell-1}}{3^{\ell}-1} \cdot \Dist{\WINNER}{z}
    = \frac{2 \cdot 3^{\ell-1}}{3^{\ell}-1} \cdot \Dist{\WINNER}{z}$.

    At least an $\alpha$ fraction of voters prefer \WINNER over $y_1$,
    and their distance to $y_1$ is at least
    $\half \Dist{\WINNER}{y_1} 
    > \frac{3^{\ell-1}}{3^{\ell}-1} \cdot \Dist{\WINNER}{z}$.
    Because the distance from $y_1$ to $z$ is at most
    $\frac{3^{\ell-1}-1}{3^{\ell}-1} \cdot \Dist{\WINNER}{z}$,
    by the triangle inequality, the distance of these voters from $z$
    is at least
    $\frac{3^{\ell-1}}{3^{\ell}-1} \cdot \Dist{\WINNER}{z}
    - \frac{3^{\ell-1}-1}{3^{\ell}-1} \cdot \Dist{\WINNER}{z}
    = \frac{1}{3^{\ell}-1} \cdot \Dist{\WINNER}{z}$.
  \end{enumerate}

  In both cases, we have thus shown that at least an $\alpha$
  fraction of voters are at distance at least
  $\frac{1}{3^{\ell}-1} \cdot \Dist{\WINNER}{z}$ from $z$.
  Thus, the cost of $z$ is at least
  $\frac{\alpha}{3^{\ell}-1} \cdot \Dist{\WINNER}{z}$.
  By the triangle inequality,

  \[ 
    \Cost{\WINNER}
    \; \leq \; \Cost{z} + \Dist{\WINNER}{z}
    \; \leq \; (1 + \frac{3^{\ell}-1}{\alpha}) \cdot \Cost{z}.
  \]
  This completes the proof of the lemma.
\end{proof}

\begin{extraproof}{Theorem~\ref{thm:upper}}
  By Lemma~\ref{lem:short-paths},
  \WINNER has a path of length at most 3 in $G$ to every candidate $x$;
  in particular, to the optimum candidate $x = \OPT$.
  Thus, by Lemma~\ref{lem:cost-ratio} with $\ell=3$,
  $\Cost{\WINNER} \leq (1 + \frac{26}{\alpha}) \cdot \Cost{\OPT}$.
  Substituting $\alpha = \frac{k}{3n}$ and bounding $1 \leq \frac{n}{k}$
  now completes the proof.
\end{extraproof}

\section{A Tight Upper Bound for Randomized Algorithms} \label{sec:randomized}
We have seen that limited communication is a serious handicap for
deterministic social choice rules,
in that all communication-bounded deterministic social choice rules
must have essentially linear distortion.
It is well known
\cite{anshelevich:postl:randomized,gross:anshelevich:xia:agree}
that this lower bound disappears for randomized social choice rules:
for example, the \emph{Random Dictatorship} mechanism,
which elects the first choice of a uniformly random voter,
has distortion slightly smaller than 3,
even though each voter only communicates her first choice.

When each voter can only communicate her first choice,
\cite{gross:anshelevich:xia:agree} proved a lower
bound of $3-\frac{2}{n}$ on the distortion of \emph{every} randomized
mechanism.
Fain et al.~\cite{fain:goel:munagala:prabhu:referee} showed that the
Random Oligarchy mechanism has a an upper bound on the distortion
almost matching the $3-\frac{2}{n}$ bound.
Here, we give a simple randomized mechanism which achieves an
expected distortion of exactly $3-\frac{2}{n}$,
thereby closing the remaining gap.
The mechanism \MECH is as follows:
\begin{itemize}
\item With probability $\frac{1}{n-1}$,
  select a candidate using the \emph{Proportional to Squares} mechanism.
  That is, for each candidate $x$, let $\nu_x$ be the fraction of voters
  who rank $x$ first.
  Select candidate $x$ with probability
  $\frac{\nu_x^2}{\sum_y \nu_y^2}$.
\item With the remaining probability $\frac{n-2}{n-1}$,
  select a candidate using the \emph{Random Dictatorship} mechanism.
  That is, choose a voter uniformly at random, and return her first choice.
  Notice that this mechanism selects candidate $x$ with probability
  exactly $\nu_x$.
\end{itemize}

We prove the following theorem:

\begin{theorem} \label{thm:mixed-upper-bound}
  The expected distortion of \MECH is at most
  $3-\frac{2}{n}$.
\end{theorem}

The proof is straightforward:
it consists of a bit of arithmetic and using Lemma~3 of
\cite{gross:anshelevich:xia:agree}, restated here in our notation.

\begin{lemma}[Lemma~3 of \cite{gross:anshelevich:xia:agree}]
\label{lem:GAX}
Let $\vc{\nu} = (\nu_x)_x$ be the vector of the fractions of voters
ranking candidate $x$ first, for all $x$.
Suppose that for every such first-place vote vector $\vc{\nu}$
and every candidate $x$,
the probability of electing $x$ under \MECH is at most
$q_{x}(\vc{\nu})$.
Then, the distortion of \MECH is at most
$1 + 2 \max_{\vc{\nu},x} (q_{x}(\vc{\nu}) \cdot \frac{1-\nu_x}{\nu_x})$. 
\end{lemma}

The main technical lemma, proved momentarily, is the following:

\begin{lemma} \label{lem:technical-bound}
For all $t \in [0,1]$, we have that
$(1-\frac{1}{n-1}) \cdot (1-t) + \frac{t (1-t)}{(n-1) t^2 + (1-t)^2}
\leq 1 - \frac{1}{n}$.
\end{lemma}

\begin{extraproof}{Theorem~\ref{thm:mixed-upper-bound}}
Let candidate $x$ be the first choice of a fraction
$\nu := \nu_x$ of voters.
The probability that $x$ is chosen under \MECH is 
\begin{align*}
(1-\frac{1}{n-1}) \cdot \nu + \frac{1}{n-1} \cdot \frac{\nu^2}{\sum_y \nu_y^2}
& \leq
(1-\frac{1}{n-1}) \cdot \nu
  + \frac{1}{n-1} \cdot \frac{\nu^2}{\nu^2
     + (n-1) \left( \frac{1-\nu}{n-1} \right)^2}
\\ & =
(1-\frac{1}{n-1}) \cdot \nu
+ \frac{\nu^2}{(n-1) \cdot \nu^2 + (1-\nu)^2}.
\end{align*}

Multiplying with the term $\frac{1-\nu}{\nu}$, we now have
\[
  (1-\frac{1}{n-1}) \cdot (1-\nu)
  + \frac{\nu \cdot (1-\nu)}{(n-1) \cdot \nu^2 + (1-\nu)^2}.
\]
By Lemma~\ref{lem:technical-bound}, this quantity is bounded by
$1-\frac{1}{n}$.
Since this bound holds for all $x$ and all $\nu_x$,
we can substitute it into Lemma~\ref{lem:GAX},
and obtain a bound of $3-\frac{2}{n}$ on the distortion, as claimed.
\end{extraproof}

\begin{extraproof}{Lemma~\ref{lem:technical-bound}}
We want to upper-bound 
$f(t) = (1-\frac{1}{n-1}) \cdot (1-t) + \frac{t (1-t)}{(n-1) t^2 + (1-t)^2}$.
First, we have that $f(0) = 1-\frac{1}{n-1}$,
and $f(1) = 0$, so the inequality holds at the extreme points.

We lower-bound the denominator $g(t) = (n-1) t^2 + (1-t)^2$ of the second term.
By setting the derivative $g'(t) = 0$,
we get that the only local extremum is a minimum at $t=\frac{1}{n}$,
where $g(1/n) = \frac{n-1}{n}$,
whereas $g(0) = 1$ and $g(1) = n-1$.
Thus, $g(t) \geq \frac{n-1}{n}$.
Substituting the lower bound on $g(t)$, we can bound 
\[
  f(t)
  \; \leq \;   (1-\frac{1}{n-1}) \cdot (1-t)
             + \frac{n \cdot t \cdot (1-t)}{n-1}.
\]
A derivative test shows that this expression has a local maximum
at $t = \frac{1}{n}$, where its value is $1-\frac{1}{n}$.
Thus, we have shown that $f(t) \leq 1-\frac{1}{n}$ for all $t \in [0,1]$.
\end{extraproof}

\section{Conclusions} \label{sec:conclusions}
As we already discussed in the introduction and
Section~\ref{sec:general-lower}, there is a gap of
$\Theta(\log n)$ in the lower bound on distortion we achieve
for $k$-entry social choice rules and more general
\PARTS-communication bounded social choice rules.
It does not appear that our techniques from
Section~\ref{sec:general-lower} can be directly generalized to produce
bounds matching the ones of Theorem~\ref{thm:simple-lower}.
Thus, if the stronger bound holds more generally,
a proof will likely require a deeper understanding of the
combinatorial structure of partitions of \ALLPERMS.
An intriguing alternative is that there may be a mechanism
in which voters communicate only $\Theta(1)$ bits of information per
candidate, but which nonetheless achieves constant distortion.
An obstacle to designing such mechanisms is that it is very unclear
how a mechanism would make use of information in which it
cannot distinguish between several very different rankings.

Throughout this article, we assumed that all voters use the same
``encoding'' in communicating with the mechanism.
For both $k$-entry social choice rules and \PARTS-communication
bounded rules, one could consider relaxing this uniformity,
although voting mechanisms which treat voters differently
\emph{a priori} are typically not widely accepted.
For $k$-entry social choice rules, our lower-bound proof can be
directly adapted to give the same lower bound so long as no voter (or
almost no voter) gets to specify which candidate she ranks last.
However, the proof does not carry over directly when some, but not
all, voters can specify their bottom-ranked candidate,
since our technique of ``sacrificing'' a candidate may come at a
higher cost to the adversary.
For \PARTS-communication bounded rules, it is much less clear how to
deal with arbitrarily differing encodings.

A further generalization would be to let voters \emph{choose} which
encoding to use, or which subset of positions to fill in.
Mechanisms allowing such a choice by the voters would have to be
considered as ``non-deterministic,'' because there is not a unique
message any more for each ranking.
This raises the issue of how a voter would determine which of many
possible messages to send.
In particular, the specific choice of message may encode additional
(e.g., cardinal) information about the voter's ranking.
It would require some subtlety to define a model to rule out the
revelation of a lot of cardinal information, while still allowing
voters non-trivial choices.

Here, we only considered single-round mechanisms.
It is well-known that in many settings,
including in the implementation of social choice rules
\cite{boutilier:rosenschein:incomplete,segal:nash-limited-communication},
multiple rounds of communication can lead to significantly
(including exponentially) lower overall communication.
Indeed,
\cite{gross:anshelevich:xia:agree,fain:goel:munagala:prabhu:referee}
studied \emph{randomized} 
multi-round voting mechanisms with the explicit goal of reducing the
required communication, while achieving low distortion.
In the case of randomized mechanisms, receiving $\log_2 n$ bits of
information from each voter is enough to achieve distortion
$3-\frac{2}{n}$ (as we showed in Section~\ref{sec:randomized} --- it
was known previously how to achieve distortion 3),
so the room for improving the required communication with multiple
rounds is limited.
However, for \emph{deterministic} mechanisms, there is potential for
significant improvement, and a natural question is whether one might
even achieve constant distortion with only $O(\log n)$ (or
$O(\text{polylog}(n))$) communication from each voter.

\subsubsection*{Acknowledgements}
The author would like to thank
Elliot Anshelevich, Yu Cheng, Shaddin Dughmi, Tyler LaBonte, Jonathan
Libgober, and Sigal Oren for useful conversations and pointers,
and anonymous reviewers for useful feedback.

\bibliographystyle{plain}
\bibliography{names,conferences,publications,bibliography}

\end{document}